\documentclass[11pt]{article}

\usepackage[margin=1in]{geometry}
\usepackage{amsmath,amssymb,amsthm,mathtools}
\usepackage{booktabs,array}
\usepackage{enumitem}
\usepackage{xcolor}
\usepackage{tikz}
\usepackage{algorithm}
\usepackage{algpseudocode}
\usepackage[colorlinks=true,linkcolor=blue!45!black,citecolor=blue!45!black,urlcolor=blue!45!black]{hyperref}

\setlist{topsep=3pt,itemsep=2pt,parsep=0pt}
\newtheorem{theorem}{Theorem}[section]
\newtheorem{lemma}[theorem]{Lemma}
\newtheorem{proposition}[theorem]{Proposition}
\newtheorem{corollary}[theorem]{Corollary}
\theoremstyle{definition}
\newtheorem{definition}[theorem]{Definition}
\theoremstyle{remark}
\newtheorem{remark}[theorem]{Remark}

\newcommand{\bcost}{\mathcal{B}}
\newcommand{\cur}[2]{C_{#1}(#2)}
\newcommand{\curlo}[2]{a_{#2}^{(#1)}}
\newcommand{\curhi}[2]{b_{#2}^{(#1)}}
\newcommand{\sh}[2]{S_{#1}(#2)}
\newcommand{\hull}[3]{H_{#3}[#1,#2]}
\newcommand{\actv}[2]{A[#1,#2]}
\newcommand{\out}{\operatorname{out}}
\newcommand{\sw}{\operatorname{sw}}
\newcommand{\wid}{\operatorname{wid}}
\newcommand{\Vsp}{\mathcal{V}}
\newcommand{\Oh}[1]{O\!\left(#1\right)}

\title{Fast Stencil Computations on a Single\\ Arbitrarily Moving Interval}
\author{Aaron Gregory\thanks{\texttt{aaron.f.gregory@stonybrook.edu}} \\
  Stony Brook University}
\date{}

\begin{document}
\maketitle

\begin{abstract}
A stencil computation repeatedly updates every cell of a grid from its
neighbours' values at the previous timestep. Simulating $T$ steps on $N$ cells
directly costs $\Theta(NT)$, and a line of work beginning with Ahmad et
al.~\cite{ahmad2021fast} reduces this by composing many timesteps into one
linear operator and applying it with a Fast Fourier Transform.

That technique needs to know which cells will still obey the same operator when
the composed step ends. In a free-boundary problem they do not: the region
governed by a given rule is determined by the solution and moves as it evolves.

We study the case of one spatial dimension, a three-point stencil with
time-varying coefficients, and a computed region that is a single interval whose
two endpoints move by arbitrary amounts at every step, revealed online. Let
$\bcost$ be the horizon plus the total variation of the boundary trajectory. We
give a schedule whose total work is $\Oh{(\bcost+N)\log T\log(N+\bcost)}$ and
whose span is $\Oh{T\log T\log(N+\bcost)}$, and we prove that the values it
computes are exact.

The best existing bound for a region that moves requires its boundary to travel at
most one cell per timestep. We drop that requirement and lose nothing by it: a
boundary obeying the requirement has $\bcost \le 3T$, so our bound stays
near-linear on every trajectory the earlier result covers. On the trajectories it
does not cover, $\bcost$ grows only by the distance the boundary actually
travels --- one jump of width $N$ costs $T + 2N$.

The reason total variation suffices is that everything the two endpoints touch over a time
window of any length lies in two intervals, one per endpoint. We also show this
cannot be relaxed: with $p$ regions the bound degrades by a factor $p$, and at
$p=\sqrt T$ there is an instance on which the work is $\Theta(T^{3/2})$ while
$\bcost+N=\Theta(T)$.

All results are machine-checked in Lean~4, apart from the classical convolution
bound of Theorem~\ref{thm:fft}, which is imported as an interface.
\end{abstract}

\section{Introduction}
\label{sec:intro}

A \emph{stencil computation} evaluates a recurrence on a spatial grid over a time
horizon: the value at a cell at time $t$ is a fixed function of the values at
that cell and its neighbours at time $t-1$. Write $N$ for the number of cells and
$T$ for the number of timesteps. Evaluating the recurrence as written --- every
cell at every step --- costs $\Theta(NT)$, and for many scientific and financial
kernels that is the dominant cost.

When the update is linear and does not vary across space, $k$ consecutive steps
compose into a single linear operator whose support has radius $k$, and applying
that operator to the whole grid is one convolution, which a Fast Fourier
Transform evaluates in $\Oh{N\log N}$ work rather than $\Theta(Nk)$. We call one
such composed application a \emph{superstep}. Trading $k$ timesteps for one
superstep is the idea behind the FFT-based stencil algorithms
of~\cite{ahmad2021fast,ahmad2022fourst,ahmad2023fast}.

\paragraph{The difficulty.}
A superstep can only be taken over a set of cells that obey the same operator for
all $k$ of the steps it spans. In a free-boundary problem this set is not known
in advance: the region governed by a given rule is determined by the solution
itself and moves as the computation proceeds. An algorithm that wants long
supersteps must therefore decide, before evaluating
them, which cells will still be governed by the same operator when the superstep
ends --- and must do so without inspecting the whole grid at every step, since
that would already cost $\Theta(NT)$.

\paragraph{The setting.}
We work in one spatial dimension. The computed region at time $t$ is a single
half-open interval $I(t) = [a_t, b_t) \subset \mathbb{Z}$, nonempty. Its two
endpoints are otherwise unconstrained: each may move by an arbitrary amount in
either direction at each step, and the trajectory is revealed online, one step at
a time. The stencil support at each step is contained in $\{-1,0,1\}$, with
coefficients that may vary with time. Cells whose dependencies reach outside the
previous region are supplied by a boundary oracle
(Section~\ref{sec:model}). The algorithm must produce every value on $I(t)$ for
every $t \le T$.

\paragraph{The parameter.}
Write $\Delta a_t = |a_t - a_{t-1}|$ and $\Delta b_t = |b_t - b_{t-1}|$, and set
\begin{equation}
  \label{eq:B}
  \bcost \;=\; T \;+\; \sum_{t=1}^{T} \bigl( \Delta a_t + \Delta b_t \bigr),
\end{equation}
the horizon plus the total variation of the boundary trajectory. The $T$ term is
forced: with coefficients that change every step, the schedule must be read, so
$\Omega(T)$ work is unavoidable even for a stationary region.

\subsection{Prior work and what it assumes}
\label{sec:prior}

The FFT-based line optimizes the cost of a homogeneous superstep.
\cite{ahmad2021fast} applies it to periodic linear stencils and
\cite{ahmad2022fourst} generates implementations;~\cite{ahmad2023fast} handles
supplied boundary conditions by a top-down recursive decomposition of the
spacetime domain. Each of these fixes the spatial partition in advance, which is
what a moving region violates. Formulations in which the update is nonlinear, so
that the set obeying a given rule is determined by the solution, are treated
in~\cite{ahmad2024fast}. \cite{companion2025} allows the operator to vary
in time and to differ between regions, with coefficients homogeneous inside each
region, and allows the regions themselves to vary in time; it introduces the
binary time-product tree we use in Section~\ref{sec:algorithm} to obtain composed
operators. We take the
binary-forking cost model and its work-efficient FFT
from~\cite{ahmad2021binary}.

Of these, only~\cite{companion2025} admits a region that moves at all, and it is
the result we compare against. Its Theorem~4.8 states that when $T \le N/B_0$,
\textsc{AperiodicSolve} computes $a_T$ in
\begin{equation}
  \label{eq:prior}
  \Oh{N\log N\log T + B_\Sigma}\ \text{work and}\
  \Oh{T\log N\log\log N}\ \text{span},
\end{equation}
where $B_t$ counts the boundary cells of every region together with the cells
governed by the boundary condition, $B_\Sigma = \sum_{t \le T} B_t$, and
$B_0$ is that count at $t=0$. It carries a \emph{regularity condition} on how a
region may move: the boundary at time $t+1$ must lie inside the region of
influence of the boundary at time $t$, so an endpoint travels at most one stencil
radius per step.

Specialized to the present setting --- one dimension, one interval, a three-point
stencil --- the regularity condition says each endpoint moves by at most one cell
per step, $B_\Sigma = \Theta(T)$, and~\eqref{eq:prior} reads
$\Oh{N\log N\log T}$ work for $T = \Oh{N}$.

\textbf{We remove the regularity condition, and the restriction on $T$ with it.}
An endpoint may move any distance at any step, and the horizon is unbounded. What
replaces the condition is the parameter: a trajectory that obeys it has
$\bcost \le 3T$, so~\eqref{eq:prior} and our bound describe the same regime there,
while a trajectory that breaks it even once leaves~\eqref{eq:prior} with nothing
to say and is charged by us at $\bcost$
(Proposition~\ref{prop:regularity}).

\subsection{Results}
\label{sec:results}

\begin{table}[t]
\centering
\small
\begin{tabular}{@{}lccl@{}}
\toprule
& Work & Span & Boundary may move \\
\midrule
Direct simulation & $\Theta(NT)$ & $\Theta(T)$ & freely \\
Bentley et al.~\cite{companion2025} & $\Oh{N\log N\log T}$ & $\Oh{T\log N\log\log N}$ & one cell per step \\
This paper & $\Oh{(\bcost{+}N)\log T\log(N{+}\bcost)}$ & $\Oh{T\log T\log(N{+}\bcost)}$ & freely \\
\bottomrule
\end{tabular}
\caption{Bounds for one dimension, one interval and a three-point stencil; row
two is Theorem~4.8 of~\cite{companion2025} specialized to that setting, and holds
only for $T = \Oh{N}$. Where row two applies, row three matches it to within a
logarithmic factor, because a boundary moving one cell per step has
$\bcost \le 3T$; where it does not, row three still applies and charges the
distance the boundary travels (Proposition~\ref{prop:regularity}).}
\label{tab:compare}
\end{table}

Our results are the following.

\begin{enumerate}[label=(\roman*)]
\item \textbf{A closed form for the cursor.} The set of cells at time $f$ that
  can be computed from time $f-w$ without consulting anything outside the region
  is an interval, and both of its endpoints are an explicit maximum and minimum
  over the intervening slices (Definition~\ref{def:cursor}). For arbitrary
  regions this set is an intersection of erosions, and a radius bound for it
  must be obtained by induction; here it is read off a formula.
\item \textbf{Exactness} (Theorem~\ref{thm:exact}), which uses neither linearity
  of the update, nor any bound on the boundary's speed, nor any relation between
  the stencil and the boundary's motion.
\item \textbf{A write count of $12(L+1)\bcost$, with no hypothesis on the
  trajectory} (Theorem~\ref{thm:writes}), where $L = \lfloor \log_2 T\rfloor$.
\item \textbf{Work and span} (Theorems~\ref{thm:work} and~\ref{thm:span}), as in
  Table~\ref{tab:compare}.
\item \textbf{A necessity result} (Proposition~\ref{prop:pieces}): with $p$
  regions the write count degrades by a factor $p$, and at $p = \sqrt{T}$ the
  bound fails by $\Theta(\sqrt T)$.
\item \textbf{An identity} (Theorem~\ref{thm:tv}) showing that the total
  variation is within a factor two of the total \emph{outward} movement plus the
  initial width, which is what justifies charging contraction as well as
  expansion.
\item \textbf{The cost of the generality} (Proposition~\ref{prop:regularity}):
  under the regularity condition $\bcost \le 3T$, and a single unconstrained step
  is enough to leave that condition behind.
\end{enumerate}

\paragraph{Scope.}
Three limitations, stated once. The stencil is three-point and the dimension is
one; Section~\ref{sec:discussion} explains why the argument does not survive
$d \ge 2$. The span is linear in $T$, because the region at time $t$ is not
determined until time $t-1$ has been computed, so the timesteps serialize; when
the trajectory is instead supplied in advance and the update is linear,
Corollary~\ref{cor:offline} gives polylogarithmic span. And what we prove about
correctness is Theorem~\ref{thm:exact} together with the covering identity of
Proposition~\ref{prop:cover}: every cell is supplied by the oracle, filled by
exactly one shell, or deferred, and each shell's values are determined by the
slice it is solved from. Algorithms~\ref{alg:solve} and~\ref{alg:step} are
presented but not formalized: we prove neither that the recursive and per-time
views coincide, nor an end-to-end statement that running the recursion produces
every requested value.

\paragraph{Roadmap.}
Section~\ref{sec:related} places the result among neighbouring lines of work.
Section~\ref{sec:model} fixes the model, the oracle, and the cost of a superstep.
Section~\ref{sec:algorithm} gives the cursor, the schedule, and a worked example;
a reader who wants only the algorithm can stop there.
Section~\ref{sec:correctness} proves exactness. Section~\ref{sec:complexity} is
the accounting, and runs in one chain: a shell is charged to a backward window,
the activity in a window is two intervals, and the windows at a given level tile.
Section~\ref{sec:necessity} shows the single-interval hypothesis cannot be
dropped, and Section~\ref{sec:discussion} discusses span and dimension.

\begin{figure}[t]
\centering
\begin{tikzpicture}[scale=0.42]
\begin{scope}
\foreach \t/\A/\Bx in {0/6/11, 1/6/12, 2/5/12, 3/5/13, 4/1/13, 5/2/13, 6/3/14, 7/3/14, 8/4/15}
  \fill[blue!11] (\A,-\t) rectangle (\Bx,-\t-1);
\draw[very thick, blue!65!black]
  (6,0)--(6,-2)--(5,-2)--(5,-4)--(1,-4)--(1,-5)--(2,-5)--(2,-6)--(3,-6)--(3,-8)--(4,-8)--(4,-9);
\draw[very thick, red!65!black]
  (11,0)--(11,-1)--(12,-1)--(12,-3)--(13,-3)--(13,-6)--(14,-6)--(14,-8)--(15,-8)--(15,-9);
\draw[draw=blue!70!black, dashed, thick] (0,0) rectangle (7,-9);
\draw[draw=red!70!black, dashed, thick] (10,0) rectangle (16,-9);
\node[blue!60!black, font=\scriptsize] at (3.5,0.7) {left strip};
\node[red!60!black, font=\scriptsize] at (13,0.7) {right strip};
\draw[<->] (-0.6,0) -- (-0.6,-9) node[midway,left,font=\scriptsize,rotate=90,anchor=south] {window};
\node[font=\scriptsize] at (8.5,-11.2) {(a) whatever the endpoints do, everything};
\node[font=\scriptsize] at (8.5,-12.2) {they touch lies in two strips};
\node[font=\scriptsize, anchor=west] at (0,-9.6) {$\underline a$};
\node[font=\scriptsize, anchor=east] at (7,-9.6) {$\overline a$};
\node[font=\scriptsize, anchor=west] at (10,-9.6) {$\underline b$};
\node[font=\scriptsize, anchor=east] at (16,-9.6) {$\overline b$};
\end{scope}
\begin{scope}[xshift=21cm, yshift=-3.7cm]
\fill[blue!11] (0,0) rectangle (16,-1.6);
\fill[orange!45] (1,0) rectangle (3,-1.6);   \fill[orange!45] (13,0) rectangle (15,-1.6);
\fill[green!35]  (3,0) rectangle (5,-1.6);   \fill[green!35]  (11,0) rectangle (13,-1.6);
\fill[purple!30] (5,0) rectangle (7,-1.6);   \fill[purple!30] (9,0) rectangle (11,-1.6);
\fill[gray!30]   (7,0) rectangle (9,-1.6);
\draw (0,0) rectangle (16,-1.6);
\foreach \x in {1,3,5,7,9,11,13,15} \draw[thin,gray] (\x,0)--(\x,-1.6);
\node[font=\scriptsize] at (2,-2.3) {$S_0$};
\node[font=\scriptsize] at (4,-2.3) {$S_1$};
\node[font=\scriptsize] at (6,-2.3) {$S_2$};
\node[font=\scriptsize] at (8,-2.3) {$\cur{2^3}{t}$};
\node[font=\scriptsize] at (14,-2.3) {$S_0$};
\node[font=\scriptsize] at (12,-2.3) {$S_1$};
\node[font=\scriptsize] at (10,-2.3) {$S_2$};
\node[font=\scriptsize, anchor=west] at (-0.2,0.6) {$I(t)$, at a time $t$ with $\nu(t)=3$};
\node[font=\scriptsize] at (8,-7.5) {(b) the shells are differences of nested};
\node[font=\scriptsize] at (8,-8.5) {cursors, hence two arms each};
\end{scope}
\end{tikzpicture}
\caption{(a) The computed region $I(t)=[a_t,b_t)$ over a window of nine steps,
including one large jump of the left endpoint. Every cell the endpoints sweep
lies in one of two strips, whose widths are the endpoints' total travel over the
window; this holds for a window of any length, which is
Lemma~\ref{lem:chaining}. (b) At a trigger time $t$, the nested cursors
$\cur{1}{t}\supset\cur{2}{t}\supset\cur{4}{t}\supset\cur{8}{t}$ cut the region into
shells, each a pair of arms, plus a deep cursor handed to the next level.}
\label{fig:main}
\end{figure}
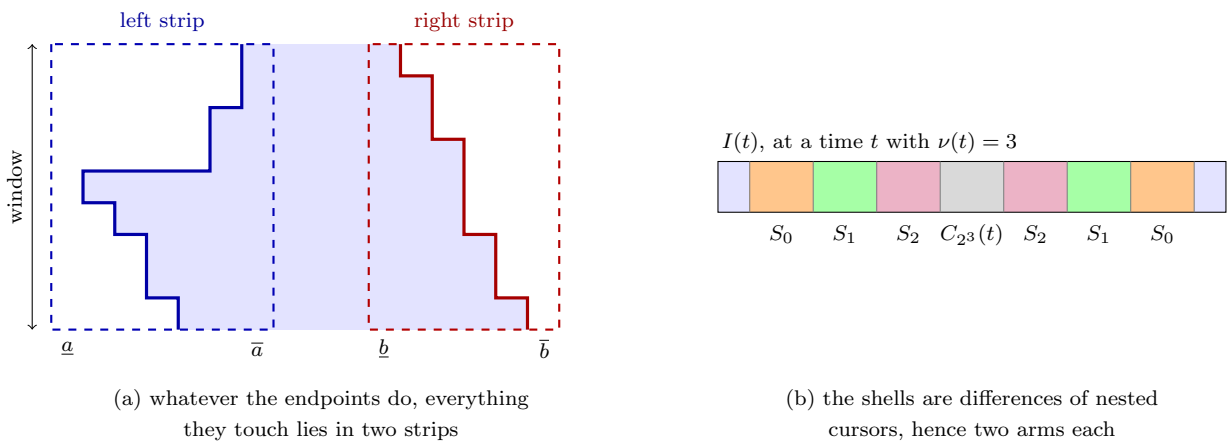

\section{Related work}
\label{sec:related}

Section~\ref{sec:prior} covered the FFT-based line this paper extends and the
bound it compares against. Three neighbouring lines are worth separating from it.

\paragraph{Direct algorithms.}
The standard stencil algorithms evaluate the recurrence as written and perform
$\Theta(NT)$ work: nested loops, tiled loops --- time skewing~\cite{wonnacott2002}
and diamond tiling~\cite{bondhugula2017} among them --- and cache-oblivious
recursive decomposition of the spacetime domain~\cite{frigo2005cache}, realized
in the Pochoir compiler~\cite{tang2011pochoir}. These differ in the order they
visit the $NT$ spacetime cells, which is what determines their locality and
parallelism, but not in how many they visit. Their generality is the reason: they
place no condition on the dimension, the support, the linearity of the update, or
the shape of the region, and so they remain the fallback whenever the hypotheses of
Section~\ref{sec:model} fail. Everything asymptotically faster buys the speedup with a structural
assumption, and the assumption is the interesting part.

\paragraph{Tracking a moving interface.}
A separate literature is devoted to representing a boundary that moves. Level-set
methods~\cite{osher1988fronts} carry the interface as the zero set of an
auxiliary field, and narrow-band variants~\cite{adalsteinsson1995fast} restrict
the update to a neighbourhood of that set so the per-step cost scales with the
interface rather than the grid. The concern there is the geometry: how to
represent a front that merges, splits or develops curvature, and how to advance it
stably. That is not the question here. We take the boundary as given --- an oracle
reports it, one step at a time --- and ask how cheaply the \emph{interior} can be
advanced once it is known. The two are complementary, and neither subsumes the
other: a narrow band still advances the interior one timestep at a time, which is
exactly the $\Theta(T)$ factor a superstep removes.

\paragraph{Cost model.}
Work and span are measured in the binary-forking
model~\cite{blelloch2019optimal,ahmad2021binary}, in which spawning $n$ threads
costs $\Theta(\log n)$ span. We use it only through Theorem~\ref{thm:fft};
nothing in Section~\ref{sec:complexity} depends on the choice.

\section{Model and cost}
\label{sec:model}

\subsection{Regions and trajectories}

\begin{definition}[Region and trajectory]
\label{def:region}
A \emph{region} is a pair of integers $a < b$, written $I = [a,b)$, with cells
$\{x \in \mathbb{Z} : a \le x < b\}$ and width $\wid(I) = b-a$. A
\emph{trajectory} assigns a region $I(t) = [a_t, b_t)$ to each $t \in
\{0,\dots,T\}$. We write $N = \wid(I(0))$.
\end{definition}

\begin{definition}[Sweep, changed cells, frontier]
\label{def:sweep}
The \emph{sweep} of an endpoint from $u$ to $v$ is $\sw(u,v) = [\min(u,v),
\max(u,v))$, which has $|u-v|$ cells. The \emph{changed cells} at step $t$ are
$D_t = \sw(a_{t-1},a_t) \cup \sw(b_{t-1},b_t)$, and the \emph{frontier} of a
region is $\Gamma(I) = \{a-1,\,a,\,b-1,\,b\}$.
\end{definition}

Sweeps are half-open so that consecutive sweeps of the same endpoint meet
exactly: $[u,v) \cup [v,w) = [u,w)$. With inclusive endpoints they would leave a
one-cell gap at each $v$, and Lemma~\ref{lem:chaining} would be false.

\begin{center}
\small
\begin{tabular}{@{}ll@{\qquad}ll@{}}
\toprule
$I(t) = [a_t,b_t)$ & the region at time $t$ & $\cur{w}{t}$ & cursor: solvable from $t-w$ \\
$N = \wid(I(0))$ & initial width & $\curlo{w}{t},\ \curhi{w}{t}$ & its endpoints \\
$D_t$ & cells that changed at $t$ & $\sh{i}{t}$ & shell filled at level $i$ \\
$\Gamma(I)$ & frontier of a region & $\actv{s}{w}$ & activity over a window \\
$\bcost$ & boundary cost, \eqref{eq:B} & $\hull{s}{w}{R}$ & its hull, dilated by $R$ \\
$\nu(t),\ L$ & $2$-adic valuation, $\lfloor\log_2 T\rfloor$ & & \\
\bottomrule
\end{tabular}
\end{center}

Objects indexed by a time are written as functions of it, with the level or
lookback as a subscript; the cursor's endpoints decorate the region's own
endpoints, and $\curlo{0}{t} = a_t$.

\subsection{The update, the oracle, and the output}

\begin{definition}[Stencil]
\label{def:stencil}
A \emph{stencil} over a value space $\Vsp$ is a family of maps
$\varphi_t : (\mathbb{Z}\to\Vsp) \to \mathbb{Z}\to\Vsp$, one per timestep, each
\emph{local}: if $g$ and $h$ agree on $\{y : |y-x|\le 1\}$ then $\varphi_t(g)(x)
= \varphi_t(h)(x)$. We write $\varphi^{(k)}_s$ for the $k$-fold iterate from time
$s$.
\end{definition}

Locality is all that correctness uses. Linearity of $\varphi_t$ is needed only
for the cost model, since it is what allows $k$ steps to be composed into one
convolution.

\begin{definition}[Boundary oracle and required output]
\label{def:oracle}
A cell of $I(t)$ whose radius-one neighbourhood is not contained in $I(t-1)$
cannot be obtained from the previous slice; its value is supplied by a
\emph{boundary oracle} at unit cost per cell. The algorithm must produce the
value at every cell of $I(T)$.
\end{definition}

Only the final slice is required. Values at intermediate times are computed where
they are needed and nowhere else: a superstep that spans many timesteps produces
its output without ever materializing the slices it crosses, and that is what
makes a bound below $\Theta(NT)$ possible at all. Requiring every intermediate
value would force $\Theta(NT)$ output on its own.

The oracle's workload is part of the cost and is charged in
Lemma~\ref{lem:oracle}: it is at most $2\bcost$ calls over the horizon.

\subsection{Cost model and the superstep primitive}
\label{sec:cost}

We use the binary-forking model~\cite{ahmad2021binary} and count arithmetic
operations in $\Vsp$ at unit cost; the bounds below are operation counts, not bit
counts. ``Exact'' throughout means exact with respect to the recurrence and the
schedule: every value produced is the value the recurrence defines, given that
each ring operation is performed exactly.

\begin{theorem}[Convolution by FFT~\cite{cooley1965,ahmad2021binary}]
\label{thm:fft}
In the binary-forking model, the convolution of two sequences of total length $n$
over a commutative ring admitting a principal $2^{\lceil\log_2 n\rceil}$-th root
of unity can be computed with $\Oh{n\log n}$ work and $\Oh{\log n}$ span.
Consequently, when the update is linear, a superstep of $k$ timesteps can be
evaluated on a block of $v$ contiguous cells with $\Oh{(v+k+1)\log(2+v+k)}$ work
and $\Oh{\log(2+v+k)}$ span.
\end{theorem}

Composing $k$ linear updates gives one operator of support radius at most $k$, so
applying it to $v$ contiguous cells is a single convolution of a sequence of
length $\Oh{v+k}$ against a kernel of length $\Oh{k}$; Theorem~\ref{thm:fft} is the only property of the update we use beyond
locality. The kernel term is carried separately because it does not follow from
the block term: a shell can be small while the level that fills it is deep.

\subsection{Why total variation is the right charge}
\label{sec:tv}

One might object that only \emph{expansion} can force work, since a cell that
leaves the region need never be revisited, and that charging contraction inflates
$\bcost$. It does, but only by a factor of two.

\begin{definition}[Outward movement]
\label{def:expansion}
$\out(T) = \sum_{t=1}^{T}\bigl((a_{t-1}-a_t)^{+} + (b_t-b_{t-1})^{+}\bigr)$ is
the total outward movement of the two endpoints.
\end{definition}

\begin{theorem}[Contraction is paid for by expansion]
\label{thm:tv}
For every trajectory and every $T$,
\[
  \sum_{t=1}^{T}\bigl(\Delta a_t + \Delta b_t\bigr)
   \;=\; 2\,\out(T) \;+\; \wid(I(0)) \;-\; \wid(I(T)),
\]
and hence $T + \out(T) \le \bcost \le 2\bigl(T + \out(T)\bigr) + N$.
\end{theorem}

\begin{proof}
Apply $|x| = 2x^{+} - x$ to each endpoint's increments and telescope; the two
telescoping sums contribute $\wid(I(T)) - \wid(I(0))$. For the upper bound use
$\wid(I(T)) \ge 1$, which holds since $a_T < b_T$; for the lower bound use
$x^{+} \le |x|$.
\end{proof}

\begin{proposition}[What the regularity condition costs]
\label{prop:regularity}
Let $\sigma \ge 1$.
\begin{enumerate}[label=(\alph*)]
\item If every step moves each endpoint by at most $\sigma$ --- the regularity
  condition of~\cite{companion2025} for a stencil of radius $\sigma$ --- then
  $\bcost \le (1+2\sigma)T$. For the three-point stencil, $\sigma = 1$ and
  $\bcost \le 3T$.
\item For every $N \ge 1$, every $T \ge 1$ and every $s < T$, the trajectory that
  holds $I(t) = [0,N)$ for $t \le s$ and $I(t) = [N,2N)$ thereafter has
  $\bcost = T + 2N$, its single largest endpoint displacement being $N$.
\end{enumerate}
\end{proposition}

\begin{proof}
(a) Each step contributes $\Delta a_t + \Delta b_t \le 2\sigma$ to the sum
in~\eqref{eq:B}, and there are $T$ of them. (b) Only the step from $s$ to $s+1$
moves either endpoint, and it moves each by exactly $N$.
\end{proof}

Part (a) is why the generality is free where the earlier result applies: on any
trajectory satisfying the regularity condition, $\bcost = \Theta(T)$, so our bound
is $\Oh{(N{+}T)\log T\log(N{+}T)}$ there --- within a logarithmic factor
of~\eqref{eq:prior}, with no restriction on $T$. Part (b) is the other side: the
trajectory sits still, jumps once by its own width, and sits still again, which
for $N > \sigma$ violates the regularity condition at exactly one step and so
falls outside~\eqref{eq:prior} entirely. Our bound still applies, and charges
$T + 2N$.

\begin{remark}[$\bcost$ is not a universal lower bound]
\label{rem:tightness}
Theorem~\ref{thm:tv} says $\bcost$ is not wasteful relative to expansion; it does
not say it is necessary. Let the region widen by $M$ and narrow again at every
step. Then every slice has at most $N+M$ cells, yet
\[
  \bcost \;=\; (M+1)\,T .
\]
So no bound in terms of the slice size and the horizon controls $\bcost$, and
$\bcost$ is not a lower bound on the work any algorithm must do. It is the right
parameter when every cell entering the region carries independent data, and not
otherwise.
\end{remark}

\section{The algorithm}
\label{sec:algorithm}

\subsection{The cursor in closed form}

\begin{definition}[Cursor]
\label{def:cursor}
For an end time $f$ and a lookback $w$,
\[
  \cur{w}{f} \;=\; \bigl[\,\curlo{w}{f},\ \curhi{w}{f}\,\bigr),
\]
\[
  \curlo{w}{f} \;=\; \max_{0\le j\le w}\bigl(a_{f-j}+j\bigr),
  \qquad
  \curhi{w}{f} \;=\; \min_{0\le j\le w}\bigl(b_{f-j}-j\bigr).
\]
We only ever use $w \le f$.
\end{definition}

Each $j$ contributes the constraint imposed by slice $f-j$: to have its $j$-step
dependency cone inside that slice, a cell must sit at least $j$ inside it. So
$\cur{w}{f}$ is exactly the set of cells at time $f$ computable from slice $f-w$
using only values interior to the region throughout. Lemma~\ref{lem:cone} gives
one direction; for the other, a cell outside the cursor falls short of the
binding constraint at some slice $f-j$, and then the cell $j$ away from it on
that side lies outside $I(f-j)$.

For example, with $w=3$ and left endpoints $a_{f-3},\dots,a_f = 0,\,5,\,4,\,4$,
the four constraints are $0+3,\ 5+2,\ 4+1,\ 4+0$, so $\curlo{3}{f} = 7$: the
slice three steps back is not the binding one.

\begin{lemma}[Nesting]
\label{lem:nesting}
If $w \le w'$ then $\cur{w'}{f} \subseteq \cur{w}{f}$; and $\cur{w}{f} \subseteq
I(f)$.
\end{lemma}

\begin{proof}
$\curlo{w}{f}$ is a maximum and $\curhi{w}{f}$ a minimum over a larger index set; the
second claim is the $j=0$ term of each.
\end{proof}

\subsection{The trigger schedule}

Let $\nu(t)$ be the $2$-adic valuation of $t$, so $\nu(12)=2$, and let $L =
\lfloor\log_2 T\rfloor$. At time $t$ the algorithm fires levels $0,\dots,\nu(t)-1$;
level $i$ fills the \emph{shell}
\[
  \sh{i}{t} \;=\; \cur{2^i}{t}\setminus\cur{2^{i+1}}{t}
\]
by one superstep from time $t-2^i$. By Lemma~\ref{lem:nesting} a shell is the difference of nested intervals, hence a
union of two intervals --- a left arm and a right arm (Figure~\ref{fig:main}b).
Each arm is evaluated by its own superstep, so a level-$i$ call is two calls to
the primitive of Theorem~\ref{thm:fft}; this doubles the call count, which the
bounds below absorb.

The schedule fires level $i$ exactly when $2^{i+1}$ divides $t$, that is, once
every $2^{i+1}$ steps. So a cell is re-derived at level $i$ only $\Oh{T/2^{i+1}}$
times over the horizon, while each level-$i$ superstep spans $2^i$ timesteps.
Summing over the $L+1$ levels is what converts timesteps into logarithmic
factors.

\begin{algorithm}[t]
\caption{$\textsc{Solve}(s,f,X)$: produce the time-$f$ values on $X$, given time $s$}
\label{alg:solve}
\begin{algorithmic}[1]
\Require $[s,f)$ an aligned dyadic interval; the time-$s$ values on the
  radius-$(f{-}s)$ neighbourhood of $X$ within $I(s)$ are available
\If{$f = s+1$}
  \For{$x \in X$}
    \State $x \gets \varphi_s(\text{slice } s)(x)$ if $x \in \cur{1}{f}$, else
      $x \gets$ oracle$(f,x)$
  \EndFor
  \State \Return
\EndIf
\If{$X \subseteq \cur{f-s}{f}$}
  \State evaluate one superstep from $s$ to $f$ onto $X$
    \Comment{Theorem~\ref{thm:fft}; no oracle value is needed, by Lemma~\ref{lem:cone}}
  \State \Return
\EndIf
\State $m \gets (s+f)/2$
\State $Y \gets \bigl(X \oplus [-(f{-}m),\,f{-}m]\bigr) \cap I(m)$
  \Comment{what time $m$ must supply}
\State $\textsc{Solve}(s,m,Y)$;\quad $\textsc{Solve}(m,f,X)$
\end{algorithmic}
\end{algorithm}

The top-level call is $\textsc{Solve}(0,T,I(T))$. The guard on line~8 is the only
place a decision is made, and Lemma~\ref{lem:cone} is what justifies it: if
$X \subseteq \cur{f-s}{f}$ then nothing influencing $X$ leaves the region during
$[s,f)$, so the ordinary stencil applies throughout and no boundary value is
consulted. Otherwise the interval is halved and the cells are recovered at finer
granularity.

Flattening the recursion gives the per-time view. At an intermediate time $t$ the
cells materialized are those that no longer aligned superstep landing at $t$ can
reach, and since an aligned superstep of length $2^i$ landing at $t$ requires
$2^i \mid t$, the longest available is $2^{\nu(t)}$. That is the shell
decomposition:

\begin{algorithm}[t]
\caption{The same schedule, viewed one timestep at a time}
\label{alg:step}
\begin{algorithmic}[1]
\State query the oracle for the endpoints $a_t, b_t$, and for the value at every
  cell of $I(t)\setminus\cur{1}{t}$
\For{$i = 0,\ \dots,\ \nu(t)-1$}
  \State evaluate one superstep from time $t-2^i$ onto $\sh{i}{t} = \cur{2^i}{t}\setminus\cur{2^{i+1}}{t}$
\EndFor
\State leave $\cur{2^{\nu(t)}}{t}$ to a longer superstep, which lands at a later
  time divisible by a higher power of two and skips time $t$ for those cells
\end{algorithmic}
\end{algorithm}

\begin{proposition}[Shell decomposition]
\label{prop:cover}
For every $t$ and every $n$, the shells $\sh{0}{t},\dots,\sh{n-1}{t}$ are pairwise
disjoint and
\[
  \cur{1}{t} \;=\; \Bigl(\bigcup_{i<n}\sh{i}{t}\Bigr)\ \cup\ \cur{2^{n}}{t}.
\]
\end{proposition}

\begin{proof}
Induction on $n$, using $\cur{2^{k}}{t} = \sh{k}{t} \cup \cur{2^{k+1}}{t}$, which is
Lemma~\ref{lem:nesting}. For disjointness, a cell of $\sh{j}{t}$ lies in
$\cur{2^{j}}{t} \subseteq \cur{2^{i+1}}{t}$ for $i<j$, while cells of $\sh{i}{t}$ lie
outside $\cur{2^{i+1}}{t}$.
\end{proof}

So the region splits into three families at each step: the cells outside
$\cur{1}{t}$, which the oracle supplies; the shells, which the chain fills, each
cell in exactly one; and the deep cursor, solvable from far enough back to be
handed to the enclosing level. Nothing is missed and nothing is filled twice.

\begin{lemma}[The oracle's workload]
\label{lem:oracle}
$\bigl|I(t)\setminus\cur{1}{t}\bigr| \le 2 + \Delta a_t + \Delta b_t$ for every
$t \ge 1$, so the oracle is called at most $2\bcost$ times over the horizon.
\end{lemma}

\begin{proof}
$\cur{1}{t} = [\max(a_t, a_{t-1}{+}1),\ \min(b_t, b_{t-1}{-}1))$ by
Definition~\ref{def:cursor}, and it lies inside $I(t)$ by
Lemma~\ref{lem:nesting}, so the difference has
$\wid(I(t)) - |\cur{1}{t}| \le (a_{t-1}{+}1-a_t)^{+} + (b_t - b_{t-1}{+}1)^{+}$
cells. Summing over $t$ and using~\eqref{eq:B} gives $2T + \sum_t(\Delta a_t +
\Delta b_t) \le 2\bcost$.
\end{proof}

\subsection{Composed operators}

A superstep over an aligned dyadic time interval needs the composed operator and
the composed support for that interval. These are maintained in a binary
time-product tree~\cite{companion2025}: a leaf for $[t,t+1)$ stores that step's
operator and support, and an internal node for $[u,w)$ with children $[u,v)$,
$[v,w)$ stores the composition of the two operators and the Minkowski sum of the
two supports. Every call in Algorithm~\ref{alg:step} uses an aligned dyadic
interval, so each requested interval is already a node of the tree and no
recomposition is needed at query time.

Because the support here is three-point, composed supports are intervals and
their Minkowski sums are computed by adding endpoints, in constant time per node.
With supports of unbounded width one must instead convolve zero-one support
masks and test positivity, to avoid mistaking coefficient cancellation for the
absence of a dependency.

\subsection{A worked example}

Take $T = 8$ and $I(0) = [4,8)$, so $N = 4$. Let the region be stationary
for three steps, then jump: $a_4 = 0$, after which it contracts by one on each of
two steps and is stationary thereafter; let $b_t$ be fixed at $8$ throughout.

\begin{center}
\begin{tabular}{@{}rrrrrl@{}}
\toprule
$t$ & $a_t$ & $b_t$ & $\Delta a_t + \Delta b_t$ & $\nu(t)$ & levels fired \\
\midrule
0 & 4 & 8 & --- & --- & --- \\
1 & 4 & 8 & 0 & 0 & none \\
2 & 4 & 8 & 0 & 1 & $0$ \\
3 & 4 & 8 & 0 & 0 & none \\
4 & 0 & 8 & 4 & 2 & $0,1$ \\
5 & 1 & 8 & 1 & 0 & none \\
6 & 2 & 8 & 1 & 1 & $0$ \\
7 & 2 & 8 & 0 & 0 & none \\
8 & 2 & 8 & 0 & 3 & $0,1,2$ \\
\bottomrule
\end{tabular}
\end{center}

Here $\bcost = 8 + 6 = 14$, against $\sum_t \wid(I(t)) = 49$ cells for direct
simulation. The jump at $t = 4$ moves the left endpoint four cells at once, so
this trajectory is outside the regularity condition and outside~\eqref{eq:prior};
it contributes $4$ to $\bcost$ and nothing else. At $t = 4$ the oracle supplies
the four cells the jump exposed plus the two at the ends; at $t = 8$ the chain
fires three levels, whose shells partition $\cur{1}{8}$ and leave $\cur{8}{8}$ to
the enclosing level.

\section{Correctness}
\label{sec:correctness}

\begin{lemma}[Cone containment]
\label{lem:cone}
If $x \in \cur{w}{f}$, $j \le w$, and $|y-x| \le j$, then $y \in I(f-j)$.
\end{lemma}

\begin{proof}
$\curlo{w}{f}$ is a maximum including the term $a_{f-j}+j$, so $x \ge a_{f-j}+j$
and $y \ge x-j \ge a_{f-j}$. Dually $x < \curhi{w}{f} \le b_{f-j}-j$ gives
$y \le x+j < b_{f-j}$.
\end{proof}

\begin{theorem}[Exactness]
\label{thm:exact}
Let $g, h : \mathbb{Z}\to\Vsp$ agree on $I(s)$. Then $\varphi^{(k)}_s(g)(x) =
\varphi^{(k)}_s(h)(x)$ for every $x \in \cur{k}{s+k}$.
\end{theorem}

\begin{proof}
By locality and induction on $k$, $\varphi^{(k)}_s(g)(x)$ depends on $g$ only
through its restriction to $\{y : |y-x| \le k\}$. By Lemma~\ref{lem:cone} with
$j=k$, that ball lies in $I(s)$, where $g$ and $h$ agree.
\end{proof}

A cell solvable from time $s$ is therefore determined by the region's values at
$s$ alone: no boundary datum can change it, and the ordinary stencil applies
throughout. The proof uses no induction along dependency chains, no bound on the
boundary's speed, and no linearity.

\section{The accounting}
\label{sec:complexity}

We prove Theorem~\ref{thm:writes}: the schedule performs $\Oh{\bcost\log T}$
writes, with no hypothesis on the trajectory. The argument is a chain of four
steps, and it is worth having in view before the notation arrives.

Each shell $\sh{i}{t}$ is written at trigger time $t$, and we charge it to the
\emph{backward window} of states $[t-2^{i+1},\,t]$. Lemma~\ref{lem:shell-local}
shows the shell lies in the two-interval hull of the window's endpoint
excursions, dilated by $2^{i+1}$. Lemma~\ref{lem:chaining} is what makes that
hull the right object: everything the endpoints touch over a window --- of any
length --- lies in those same two intervals, each no wider than the window's
total variation. Fattening two
intervals by a radius $R$ adds $4R$ cells; fattening $2p$ of them adds $4pR$, and
that factor $p$ is the whole of Section~\ref{sec:necessity}. Finally
Lemma~\ref{lem:tiling} shows the windows of a given level are disjoint, so each
level costs one copy of the total variation, and there are $L+1$ levels.

\subsection{Activity over a window is two intervals}

Fix a window of states $s,\,s+1,\,\dots,\,s+w$ and let $\actv{s}{w}$ be all
activity over it, namely $\Gamma(I(s))$ together with $D_{u+1} \cup
\Gamma(I(u+1))$ for $s \le u < s+w$. Write $\underline a, \overline a$ for the
least and greatest of $a_s,\dots,a_{s+w}$, and $\underline b, \overline b$
likewise. For $R \ge 0$ put
\[
  \hull{s}{w}{R} \;=\; \bigl[\underline a - 1 - R,\ \overline a + R\bigr]
                \ \cup\ \bigl[\underline b - 1 - R,\ \overline b + R\bigr].
\]

\begin{lemma}[Chaining]
\label{lem:chaining}
$\actv{s}{w} \subseteq \hull{s}{w}{0}$, and every cell within distance $R$ of
$\actv{s}{w}$ lies in $\hull{s}{w}{R}$.
\end{lemma}

\begin{proof}
Induction on $w$. For every $u$ in the window the frontier cells $a_u-1, a_u$ lie
in the first interval and $b_u-1, b_u$ in the second, by definition of the
extremes. The changed cells $D_{u+1}$ split into a left sweep between $a_u$ and
$a_{u+1}$ and a right sweep between $b_u$ and $b_{u+1}$; each has both endpoints
within the corresponding extremes, and being half-open it meets the previous
step's sweep rather than starting a new component. A point within $R$ of a member
of an interval lies in that interval's $R$-dilation.
\end{proof}

The window is indexed by its states rather than by its activity times because
activity at time $u$ refers to $I(u-1)$, so a window of $w$ activity times spans
$w+1$ states.

\begin{lemma}[The hull is small]
\label{lem:hull-card}
$\bigl|\hull{s}{w}{R}\bigr| \le 4 + 2\sum_{u=s+1}^{s+w}(\Delta a_u + \Delta b_u) +
4R$.
\end{lemma}

\begin{proof}
Two intervals, each of size (extreme spread) $+\,2R+2$. Each endpoint's spread
over a window is at most the window's total variation, by induction on the window
length: adjoining a state moves the running maximum and minimum apart by at most
the new increment. Two intervals, each charged the whole variation, give the
factor $2$.
\end{proof}

The two summands are the two terms of the complexity: the variation is charged to
$\bcost$, and the $4R$ to the schedule. The coefficient of $R$ is $4$ --- twice
the number of intervals --- and, crucially, it does not grow with $w$.

\subsection{The shells are local}

\begin{lemma}[Shell locality]
\label{lem:shell-local}
If $w \le w' \le f$ then $\cur{w}{f}\setminus\cur{w'}{f} \subseteq
\hull{f-w'}{w'}{w'}$. In particular $\sh{i}{t} \subseteq
\hull{t-2^{i+1}}{2^{i+1}}{2\cdot 2^{i}}$.
\end{lemma}

\begin{proof}
A cell of the difference fails the longer-lookback cursor on one side. If on the
left, then $x < \curlo{w'}{f} = a_{f-j}+j$ for some $j \le w'$, so $x \le
\overline a + w'$; and $x \ge \curlo{w}{f} \ge a_f \ge \underline a$. The right
side is dual.
\end{proof}

Without a closed form for the cursor, a bound of this shape has to be obtained by
induction along dependency chains, and the constant it yields is larger. Here the
radius is the window length itself, read off Definition~\ref{def:cursor}.

\subsection{Charging}

\begin{lemma}[Level windows tile]
\label{lem:tiling}
Fix $\ell \ge 1$. As $t$ ranges over the multiples of $2^{\ell}$ in $[1,T]$, the
backward windows $[t-2^{\ell},\,t]$ have pairwise disjoint increment sets, so
\[
  \sum_{t} \ \sum_{u \in [t-2^\ell,\,t)} (\Delta a_{u+1} + \Delta b_{u+1})
    \;\le\; \bcost - T.
\]
\end{lemma}

\begin{proof}
Distinct multiples of $2^\ell$ differ by at least $2^\ell$, so the half-open
increment ranges $[t-2^\ell, t)$ are disjoint and contained in $[0,T)$.
\end{proof}

Charging by trigger time rather than by written slice is what keeps the window
backward-looking. A chronological algorithm on an online trajectory cannot
justify a write by activity that has not happened yet, so a window extending
forward from the written slice would be unavailable; charging by the pair (trigger
time, level) makes the direction explicit, and the windows then tile instead of
merely overlapping boundedly.

\begin{lemma}[Total band radius]
\label{lem:radii}
$\sum_{t=1}^{T} 2^{\nu(t)} \le (L+1)\,T$.
\end{lemma}

\begin{proof}
$2^{\nu(t)} \le \sum_{j\le L}[\,2^j \mid t\,]\,2^j$. Exchanging the order of
summation gives $\sum_{j\le L} 2^j\lfloor T/2^j\rfloor \le (L+1)T$.
\end{proof}

\begin{theorem}[Maintenance write count]
\label{thm:writes}
The cells materialized at intermediate times satisfy
$\displaystyle\sum_{t=1}^{T}\ \sum_{i<\nu(t)} \bigl|\sh{i}{t}\bigr| \;\le\;
12\,(L+1)\,\bcost$, with no hypothesis on the trajectory.
\end{theorem}

\begin{proof}
By Lemma~\ref{lem:shell-local} each shell lies in a hull, whose size
Lemma~\ref{lem:hull-card} bounds by
\[
  4 \;+\; 2\!\!\sum_{u \in (t-2^{i+1},\,t]}\!\!(\Delta a_u + \Delta b_u)
    \;+\; 8\cdot 2^{i}.
\]
Sum over $i < \nu(t)$ and $t \le T$. The constant terms give
$4\sum_t \nu(t) \le 4(L+1)T$. The variation terms give at most
$2(L+1)(\bcost - T)$: by Lemma~\ref{lem:tiling} each level contributes one copy,
and there are $L+1$ levels. The radius terms give at most $8(L+1)T$ by
Lemma~\ref{lem:radii}. The two terms proportional to $T$ contribute
$(4+8)(L+1)T$ and the variation term contributes $2(L+1)(\bcost-T)$; majorizing
both coefficients by $12$ and using $T + (\bcost - T) = \bcost$ gives
$12(L+1)\bcost$.
\end{proof}

Theorem~\ref{thm:writes} counts the cells materialized at intermediate times.
Two families sit outside it and are bounded separately: the oracle's cells, at
most $2\bcost$ by Lemma~\ref{lem:oracle}, and the output slice $I(T)$, at most $N + \bcost$ by Lemma~\ref{lem:block}. What
it costs to \emph{produce} these values is the superstep primitive, charged
next.

\subsection{Work and span}

\begin{lemma}[Block size]
\label{lem:block}
Every block the schedule evaluates has at most $N + \bcost$ cells.
\end{lemma}

\begin{proof}
A shell lies in $\cur{2^i}{t} \subseteq I(t)$ by Lemma~\ref{lem:nesting}, and
$\wid(I(t)) \le \wid(I(0)) + \sum_{u\le t}(\Delta a_u + \Delta b_u) \le N +
\bcost$ by telescoping.
\end{proof}

\begin{theorem}[Work]
\label{thm:work}
The schedule's total work is
$\Oh{\bigl(\bcost + N\bigr)\,(L+1)\,\log(N{+}\bcost)}$.
\end{theorem}

\begin{proof}
By Theorem~\ref{thm:fft} a level-$i$ call on a block of $v$ cells costs
$\Oh{(v + 2^{i} + 1)\log(2+v+2^{i})}$. Lemma~\ref{lem:block} bounds $v$ by
$N+\bcost$, and $2^{i} < 2^{\nu(t)} \le t \le T \le \bcost$, so every logarithm is
$\Oh{\log(N+\bcost)}$. Summing over the intermediate calls, the $v$ terms give the
write count of Theorem~\ref{thm:writes}; the kernel terms give
$\sum_t\sum_{i<\nu(t)} 2^{i} \le \sum_t 2^{\nu(t)} \le (L+1)T$ by
Lemma~\ref{lem:radii}; and the $+1$ terms give the number of calls,
$\sum_t \nu(t) \le (L+1)T$. Both are at most $(L+1)\bcost$ since $T\le\bcost$. The
output slice adds one superstep on at most $N+\bcost$ cells, and the oracle
contributes $\Oh{\bcost}$ by Lemma~\ref{lem:oracle}.
\end{proof}

\begin{theorem}[Span]
\label{thm:span}
The schedule's total span is $\Oh{\log(N{+}\bcost)\cdot(L+1)\cdot T}$.
\end{theorem}

\begin{proof}
Each call contributes span $\Oh{\log(N{+}\bcost)}$ by Theorem~\ref{thm:fft},
using the same bounds on $v$ and $2^{i}$ as above, and time $t$ makes $\nu(t)$
calls; sum using $\sum_t \nu(t) \le (L+1)T$.
\end{proof}

The span is linear in $T$. That is not an artefact of the analysis: the region at
time $t$ is not determined until time $t-1$ has been computed, so the timesteps
serialize. We do not know a schedule that does better, and we do not prove a
matching lower bound. What the bound does not contain is $\bcost$ --- an arbitrarily
large jump costs work, not depth. When the trajectory is known in advance the
serialization disappears:

\begin{corollary}[Depth of a balanced schedule]
\label{cor:offline}
A schedule organized as a balanced binary tree over $2^{h}$ supersteps, each of
depth $\delta$, has depth $(h+1)\,\delta$.
\end{corollary}

\begin{proof}
Every root-to-leaf path costs the same, so the depth of a node is the maximum of
its children's depths plus $\delta$; induction on the tree gives
$(\mathrm{height}+1)\delta$, and the balanced tree over $2^{h}$ leaves has height
$h$.
\end{proof}

When the trajectory is supplied in advance and the update is linear, operator
composition is associative over time, so the supersteps may be arranged as such a
tree rather than as a chronological loop. With $h = \lceil\log_2 T\rceil$ and
$\delta = \Oh{\log(N{+}\bcost)}$ this gives span $\Oh{\log T\log(N{+}\bcost)}$ at
the same work.

With $L = \Oh{\log T}$ the two theorems read
\[
  \text{work } \Oh{(\bcost+N)\log T\log(N+\bcost)},
  \qquad
  \text{span } \Oh{T\log T\log(N+\bcost)},
\]
as in Table~\ref{tab:compare}.

\section{A single interval is necessary}
\label{sec:necessity}

Lemma~\ref{lem:chaining} gives two intervals for one region and $2p$ for $p$ of
them, so the radius term of Lemma~\ref{lem:hull-card} becomes $4pR$ and
the same argument would give $\Oh{p\,(L+1)\bcost}$ for
Theorem~\ref{thm:writes}. That degradation is attained.

\begin{proposition}[The radius is paid once per component]
\label{prop:pieces}
Let $R \ge 0$.
\begin{enumerate}[label=(\alph*)]
\item If $p$ centers are pairwise at least $2R+1$ apart, their radius-$R$
  dilation has exactly $p\,(2R+1)$ cells.
\item If every center lies in $[u,v]$, the dilation has at most $v-u+2R+1$ cells,
  whatever $p$ is.
\item With $p = 2^{n}$ such components of radius $R = 2^{n}$, fired at each of
  $2^{n}$ trigger times, the dilation total is at least $2\cdot 8^{n}$; for
  $T = 4^{n}$ this is $2\,T^{3/2}$.
\end{enumerate}
\end{proposition}

\begin{proof}
(a) The balls around separated centers are pairwise disjoint and each has $2R+1$
cells. (b) The dilation is contained in $[u-R,\,v+R]$. (c) Multiply: $2^{n}$
trigger times, $2^{n}$ components, $2\cdot 2^{n}+1$ cells each, and
$2^{n}\cdot 2^{n}\cdot(2\cdot 2^{n}+1) = 2\cdot 8^{n} + 4^{n}$.
\end{proof}

Parts (a) and (b) are the whole of it: a single interval pays the radius once,
$p$ separated components pay it $p$ times, and Lemma~\ref{lem:hull-card}'s
coefficient $4R$ becomes $4pR$. Part (c) turns that factor into a $T^{3/2}$
total, against a boundary cost that such a family keeps at $\Theta(T)$ --- the
single-interval bound of Theorem~\ref{thm:writes} would give $\Oh{T\log T}$ for
the same data, so the gap is $\Theta(\sqrt{T})$.

No polylogarithmic factor absorbs $\Theta(\sqrt T)$. Note that such a family can
keep its \emph{mean} activity at $\Oh{1}$, so bounding average activity does not
rescue the bound. What fails is that nothing ties the number of components to the
$2$-adic valuation, and a single interval forecloses that by having boundedly many
components in every window rather than boundedly many cells.

\section{Discussion}
\label{sec:discussion}

\paragraph{Higher dimensions.}
Lemma~\ref{lem:chaining} replaces a union of swept bands by two intervals, which
is available only in one dimension. The corresponding statement for $d \ge 2$
would convert a $(d-1)$-dimensional surface into a $d$-dimensional volume, and the
component count that makes Lemma~\ref{lem:hull-card} work has no analogue there.
We do not know whether the obstruction is avoidable.

\paragraph{Span.}
Corollary~\ref{cor:offline} settles the offline case. In the online case we
believe the right bound has the shape $\tilde{O}(m)$, where $m$ is the number of
timesteps at which the boundary actually moves, since the dependency chain
through the boundary decisions has depth $m$; obtaining it appears to require
speculating over blocks and verifying, which is sound only under a locality
hypothesis on the criterion that moves the boundary. We prove neither that upper
bound nor a matching lower bound.

\paragraph{When the bound is strong.}
The work is $\Oh{(\bcost+N)\log T\log(N{+}\bcost)}$, so it is near-linear exactly
when the boundary's total travel is comparable to the horizon. A boundary that
moves $\Oh{1}$ cells per step on average has $\bcost = \Oh{T}$ and total work
$\Oh{(N+T)\log T\log(N{+}T)}$ --- whatever its largest single step, since one
arbitrarily long jump is absorbed into the sum. The bound is weak when the total
travel is superlinear in $T$, and Remark~\ref{rem:tightness} shows that in that
regime it can also be far from optimal.

\paragraph{Acknowledgements.}
This work was supported in part by NSF grant CCF-2318633.

\begingroup
\small
\bibliographystyle{abbrv}
\bibliography{references}

\begin{thebibliography}{10}

\bibitem{adalsteinsson1995fast}
D.~Adalsteinsson and J.~A. Sethian.
\newblock A fast level set method for propagating interfaces.
\newblock {\em Journal of Computational Physics}, 118(2):269--277, 1995.

\bibitem{ahmad2024fast}
Z.~Ahmad, R.~Browne, R.~Chowdhury, R.~Das, Y.~Huang, and Y.~Zhu.
\newblock Fast {American} option pricing using nonlinear stencils.
\newblock In {\em Proceedings of the 29th ACM SIGPLAN Annual Symposium on
  Principles and Practice of Parallel Programming}, pages 316--332, 2024.

\bibitem{ahmad2021binary}
Z.~Ahmad, R.~Chowdhury, R.~Das, P.~Ganapathi, A.~Gregory, and M.~M. Javanmard.
\newblock Low-span parallel algorithms for the {Binary-Forking} model.
\newblock In {\em Proceedings of the 33rd ACM Symposium on Parallelism in
  Algorithms and Architectures}, pages 22--34, 2021.

\bibitem{ahmad2021fast}
Z.~Ahmad, R.~Chowdhury, R.~Das, P.~Ganapathi, A.~Gregory, and Y.~Zhu.
\newblock Fast stencil computations using {Fast Fourier Transforms}.
\newblock In {\em Proceedings of the 33rd ACM Symposium on Parallelism in
  Algorithms and Architectures}, pages 8--21, 2021.

\bibitem{ahmad2023fast}
Z.~Ahmad, R.~Chowdhury, R.~Das, P.~Ganapathi, A.~Gregory, and Y.~Zhu.
\newblock A fast algorithm for aperiodic linear stencil computation using {Fast
  Fourier Transforms}.
\newblock {\em ACM Transactions on Parallel Computing}, 10(4):1--34, 2023.

\bibitem{ahmad2022fourst}
Z.~Ahmad, M.~M. Javanmard, G.~Croisdale, A.~Gregory, P.~Ganapathi, L.-N.
  Pouchet, and R.~Chowdhury.
\newblock Fourst: A code generator for {FFT}-based fast stencil computations.
\newblock In {\em 2022 IEEE International Symposium on Performance Analysis of
  Systems and Software (ISPASS)}, pages 99--108. IEEE, 2022.

\bibitem{companion2025}
R.~Bentley, R.~Chowdhury, A.~Gregory, and M.~Santomauro.
\newblock Applying {Fast Fourier Transforms} to accelerate spatially and
  temporally inhomogeneous stencil computations.
\newblock In {\em Proceedings of the 37th ACM Symposium on Parallelism in
  Algorithms and Architectures (SPAA)}, pages 17--33, 2025.

\bibitem{blelloch2019optimal}
G.~E. Blelloch, J.~T. Fineman, Y.~Gu, and Y.~Sun.
\newblock Optimal parallel algorithms in the binary-forking model.
\newblock {\em arXiv preprint arXiv:1903.04650}, 2019.

\bibitem{bondhugula2017}
U.~Bondhugula, V.~Bandishti, and I.~Pananilath.
\newblock Diamond tiling: Tiling techniques to maximize parallelism for stencil
  computations.
\newblock {\em IEEE Transactions on Parallel and Distributed Systems},
  28(5):1285--1298, 2017.

\bibitem{cooley1965}
J.~W. Cooley and J.~W. Tukey.
\newblock An algorithm for the machine calculation of complex {Fourier} series.
\newblock {\em Mathematics of Computation}, 19(90):297--301, 1965.

\bibitem{frigo2005cache}
M.~Frigo and V.~Strumpen.
\newblock Cache oblivious stencil computations.
\newblock In {\em Proceedings of the 19th International Conference on
  Supercomputing}, pages 361--366, 2005.

\bibitem{osher1988fronts}
S.~Osher and J.~A. Sethian.
\newblock Fronts propagating with curvature-dependent speed: Algorithms based
  on {Hamilton--Jacobi} formulations.
\newblock {\em Journal of Computational Physics}, 79(1):12--49, 1988.

\bibitem{tang2011pochoir}
Y.~Tang, R.~A. Chowdhury, B.~C. Kuszmaul, C.-K. Luk, and C.~E. Leiserson.
\newblock The {Pochoir} stencil compiler.
\newblock In {\em Proceedings of the 23rd ACM Symposium on Parallelism in
  Algorithms and Architectures}, pages 117--128, 2011.

\bibitem{wonnacott2002}
D.~Wonnacott.
\newblock Achieving scalable locality with time skewing.
\newblock {\em International Journal of Parallel Programming}, 30(3):181--221,
  2002.

\end{thebibliography}
\endgroup

\end{document}